\documentclass[runningheads]{llncs}
\usepackage{graphicx}
\usepackage{booktabs}
\usepackage{bm}
\usepackage{amsmath}
\usepackage{amssymb}
\usepackage[square,numbers]{natbib}
\usepackage{float}
\usepackage{subcaption}
\usepackage[a4paper,top=3cm,bottom=2cm,left=3cm,right=3cm,marginparwidth=1.75cm]{geometry}

\begin{document}
	%
	\title{Quantifying Filter Bubbles: Analyzing Surprise in Elections}
	%
	%
	\author{Massand Sagar Sunil \orcidID{0000-0002-1095-6853}\and
		Swaprava Nath \orcidID{0000-0001-8309-5006}}
	%
	%
	\institute{Indian Institute of Technology, Kanpur \\ 
		\email{\{smassand,swaprava\}@cse.iitk.ac.in}}
	\maketitle              
	\begin{abstract}
		%
		This work analyses surprising elections, and attempts to quantify the notion of surprise in elections. A voter is surprised if their estimate of the winner (assumed to be based on a combination of the preferences of their social connections and popular media predictions) is different from the true winner. A voter's social connections are assumed to consist of contacts on social media and geographically proximate people. We propose a simple mathematical model for combining the global information (traditional media) as well as the local information (local neighbourhood) of a voter in the case of a two candidate election. We show that an unbiased, influential media can nullify the effect of filter bubbles and result in a less surprised populace. Surprisingly, an influential media source biased towards the winners of the election also results in a less surprised populace. Our model shows that elections will be unsurprising for \emph{all} of the voters with a high probability under certain assumptions on the social connection model in the presence of an influential, unbiased traditional media source. Our experiments with the UK-EU referendum (popularly known as Brexit) dataset support our theoretical predictions. Since surprising elections can lead to significant economic movements, it is a worthwhile endeavour to figure out the causes of surprising elections.

	\end{abstract}
	\section{Introduction}
	There have been a spate of recent elections where the election results have been surprising, in terms of polling data as well as voters' impressions of how the election would turn out.
	The 2016 US presidential elections and the Brexit referendum are prominent examples of this phenomenon. The economic activity after the election result in the case of Brexit confirmed that the election result was not expected by a large fraction of the voters (\citet{ramiah2017sectoral}). Surprising results can also lead to overreactions in the stock markets (\citet{de1985stockOverreaction}), and reduced liquidity in the markets just before the results (\citet{cox2017political}). 
	Thus, these surprising outcomes can possibly lead to economic instability, even triggering financial collapses around the world. 
	Additionally, widespread perception that a candidate is winning could lead to depressed voter turnout for them and sway independent voters against them, leading to a change in the election outcome.  
	Thus, it becomes imperative to understand the structures governing individual predictions, and how they could be improved to ensure that voter surprise is minimised.
	
	One of the reasons for such surprising elections is traditional media and polling websites getting their predictions of the election wrong. In the case of the 2016 US presidential elections and the Brexit referendum, popular traditional media and polling sites were predicting wins for Clinton and the Remain camp respectively. However, this was not the first time that polls were erroneous, with polling errors occurring even in the 2015 UK elections. Additionally, later studies (\citet{jennings2018election}, \citet{blumenthalUS2016evaluation}) showed that the magnitude of polling errors in these elections had been largely consistent with historical polling errors. What had changed, however, was the influence of social media in people's lives. 
	
	In a world increasingly connected by social media, and with the decline of traditional news media, news is increasingly consumed via social media. However, on social media, people are more likely to see content tailored to their liking and more likely to connect to people who share their views (\citet{pariser2011filter}). Thus, it becomes straightforward for a user to believe that most people agree with their views and support their candidate. This state of intellectual isolation (due to websites showing information that they guess the user would want to see) is known as a filter bubble (\citet{pariser2011filter2}). The term was coined by Eli Pariser in his 2010 book, in which he explores the issues with personalization of search. Such filter bubbles are believed to be part of the reason for the surprising results in the 2016 US presidential elections and the Brexit referendum (\citet{difranzo2017filter}). In cases where the predictions from both traditional and social media are incorrect, we would have a large fraction of the voting population being surprised. Hence, in this work, we attempt to take into account the influence of both traditional and social media, and model a voter's prediction framework as a combination of traditional media predictions and the voter's social network structure.
	
	\subsection{Related work} 

	The notions of surprise and suspense are defined in a recent work by \citet{ely2015suspense} in a setting with a fixed state space and a finite number of periods dependent on the belief martingale. In their particular model, the surprise is considerable in a particular period if the current belief is far from last period's belief, while the suspense generated is significant if the variance of next period's beliefs is large. The motivation for this model comes from elections, sports and mystery novels. There are also some interesting works on estimating the margin of victory, by \citet{dey2015estimating}, which can serve as a basis for extensions to the model that we have proposed in this work.

	Social media, and its connection with elections has received a lot of attention in the recent past. Some of these works focus on the impact of social media on elections (\citet{metaxas2012social},\citet{broersma2012social}), polling using social media (\citet{o2010tweets},\citet{cummings2010needs},\citet{tumasjan2010predicting}) and the spread of fake news using social media(\citet{difranzo2017filter},\citet{maheshwari2016fake},\citet{allcott2017social}). \citet{difranzo2017filter} describe how the 2016 US presidential campaigns propogated fake news. \citet{metaxas2012social} describe methods to manipulate social media and search results in favour of the candidates,while \citet{broersma2012social} analyse the impact of social media as a cheaper news source in the election cycle. 
	
	There has also been some recent work regarding surprise in elections. \citet{scheve1999electoral} analyse U.S. midterm elections, and argue that the more surprised voters are about the outcome of the presidential election, the less likely they are to be voting for the president's party in the midterm elections. On the other hand, \citet{dey2018surprise} look at an individual boolean notion of surprise, i.e. a voter is surprised if their predicted candidate does not win, and unsurprised if their predicted candidate wins. In the two candidate model, voters are divided into two groups based on their preferred candidates. A voter's connection probability to another voter in the same group ($p$) is greater than her connection probability to a voter in the opposite group($q$). The voter's prediction ($Z$) is formulated from her connections as well as an unconstrained belief about her connection probabilities to voters in the same group ($\hat{p}_{ii}$) as well as the voters in the opposite group. ($\hat{p}_{ij}$) as follows:
	\begin{equation}
Z = \begin{cases}
i & \quad \text{if } \frac{C_{ii}}{\hat{p}_{ii}} + 1 - \frac{C_{ij}}{\hat{p}_{ij}} > 0 \\
j & \quad \text{otherwise}
\end{cases}
	\end{equation}
	where $C_{ii}$ is the number of connections the voter has to other voters of his own class, while $C_{ij}$ is the number of connections the voter has to voters of the other class. 
		\subsection{Our contribution}

	We propose a voter prediction framework which combines the information obtained from the voter's social network structure (local information), and traditional media (global information), which includes news sources and polling websites. From the social network structure, our model uses the number of voters of both the candidates who are the voter's connections.
	In contrast to \citet{dey2018surprise}, which assumes a symmetric connection model between voters, considering the asymmetric nature of modern social media, we take the network structure to be a directed graph. A voter $j$ is a connection of voter $i$ if there is a directed edge from $i$ to $j$. This asymmetric structure corresponds closely to the asymmetric follower structure of most major social networks instead of the symmetric friend structure. It also allows a more realistic modelling of the outsized influence that a few power users have in any social network. 
	
	Our primary focus in this work is two candidate elections. 
	Given a static distribution of voters for the two candidates, and under the assumption that there are a sufficiently large number of voters, we look at the combined influence of traditional and social media on the accuracy of voter predictions about the election. Under certain assumptions on the social network structure, 
	our model predicts that all voters will be unsurprised with a high probability in the presence of an influential media source which is either unbiased or biased towards the winning candidate. We also predict that the voters of the winning candidate will not be surprised unless there is a significant media bias towards the losing candidate. 
	
	In the presence of an uninfluential traditional media source, our theoretical results show that voter surprise is the same as in the case of a complete absence of traditional media, for elections with a sufficiently large number of voters. To analyse elections with a smaller voting population, we conduct experiments on small sub-samples of the Brexit dataset. Our experimental results show that an uninfluential media source which is neutral or biased towards the winning candidate can reduce the surprise of the minority voters. 
	
	\section{Model}
	There are several components in our model which need to be understood to understand the notion of surprise in our model. We start with the voter model, where we divide the voters into classes based on the their preferred candidates. Since we are looking at two candidate elections with standard, anonymous voting rules, all of the standard voting rules reduce to plurality voting. We then move on to the voter connection model, which describes the asymmetric social network structure used in our model, and how voters are able to observe the preferred candidates of their connections. Voters' winner perception model combines the influence of traditional media and social media, and is followed by the traditional media model, which explains the modelling  of influential and uninfluential traditional media sources. 
	\subsection{Voter model}
    Let $N= \{1,2,...,n\}$ be the set of voters and $M = \{a_1, a_2\}$ be the set of candidates. The voters have a linear ordering over each of the candidates (i.e. they have ordinal preferences which are transitive, anti-symmetric and complete), with each linear ordering defining a class. Thus, in this case, the voters are divided into $2$ disjoint classes. Let $H_1$ denote the set of voters who prefer $a_1$ over $a_2$, and $H_2$ denote the set of voters who prefer $a_2$ over $a_1$. Since we will be looking exclusively at standard, anonymous voting rules and there are only two candidates, the outcome of the election is only decided by the number of voters in each class, rather than the identities of the individual members. Thus, we can create a vector of voters $\overrightarrow{N}$ = ($N_1, N_2$)(where $N_i$ denotes the number of voters belonging to class $H_i$), which suffices to decide the outcome of any election (with ties being broken in favour of candidate $a_2$, without loss of generality). In our model, we assume that $N$ is fixed instead of being sampled from a pre-defined probability distribution for each of the classes taken in most models (e.g. \citet{dey2018surprise}). The mapping from voters to their respective classes is denoted by $ \sigma: N \rightarrow \{1,2\}$. Without loss of generality, we take the specific instance where $|H_1| = n(\frac{1}{2} + \epsilon)$ and $|H_2| = n(\frac{1}{2} - \epsilon)$, with the constant $\epsilon \in (0, \frac{1}{2})$. Note that, this implies that candidate $a_1$ wins the election in our analysis. The voters in $H_1$ are referred to as the majority voters, and the voters in $H_2$ are referred to as the minority voters.

	\subsection{Voter connection model}
	Our connection model is a slight modification of the stochastic block model (\citet{mossel2012stochastic}). Similar to the stochastic block model, we have a $2 \times 2$ matrix, $P = [p_{ij}]$, $i,j \in \{1,2\}$. $p_{ij}$ denotes the probability of a player of class $H_i$ having a \textit{directed} edge to a player of class $H_j$. Additionally, $p_{11} = p_{22} > p_{12} = p_{21}$. By overloading notation, we denote the intra-class connection probability by $p = p_{11} = p_{22}$, and the inter-class connection probability by $q = p_{12} = p_{21}$. Note that, since there are filter bubbles in most popular social networks, the intra-class connection probability $p$ is justifiably taken to be  larger than the inter-class connection probability $q$. The graph thus formed is denoted by $G = (N, E)$. A directed edge $(i,j)$ from voter $i$ to voter $j$ denotes that voter $i$ can \textit{observe} voter $j$'s preferred candidate, but not vice-versa. If, however, both $(i,j)$ and $(j,i)$ exist in the graph, both $i$ and $j$ can view each other's preferred candidate. This asymmetric, directed graph structure is used as it corresponds closely to the follower structure of most popular social networks. For a voter $i$, we denote its set of neighbours/connections with $N_i = \{j|(i,j) \in E\}$, with $N_{i,1} = N_i \cap H_1$ and $N_{i,2} = N_i \cap H_2$.
	
	\subsection{Voters' winner perception model}
	For each voter $i$, we denote $i$'s predicted winner by $Z_i$. Each voter $i$ assumes a binomial distribution for the voters and tries to estimate the unknown parameter $q_1$, the probability of a voter preferring candidate $a_1$ over $a_2$. Voter $i$'s estimate of $q_1$ is denoted by $q_{1,i}$ and is computed by using the maximum a posteriori (MAP) estimate, with the prior being a Beta distribution, and the data (or evidence) being the voters that they observe. 
	 The Beta prior denotes the global information, which can be assumed to be obtained from media sources such as news (print and electronic) as well as polls. The data denotes the local information, which is obtained from the voter's social network graph, specifically, the number of voters of $a_1$ and $a_2$ who are their connections. The Beta prior has two parameters $\alpha$ and $\beta$, and denotes the likelihood of different $q_1$ values before observing any data. As each voter $i$ computes $q_{1,i}$ as the MAP estimate, 
	\begin{equation}
	\label{probEq}
	q_{1,i} = \frac{\alpha + N_{i,1} + I(i \in H_1) - 1}{\alpha + \beta + N_i - 1} 
	\end{equation}
\begin{equation}
\label{predEq}
Z_i = \begin{cases}
a_1 &\quad\text{if } q_{1,i} > 0.5 \\
a_2 & \quad\text{otherwise}
\end{cases} 
\end{equation}

By equations \ref{probEq} and \ref{predEq},
\begin{equation}
\label{jointEq}
Z_i = \begin{cases}
a_1 & \quad\text{if } \alpha + 2 I(i \in H_1) - \beta > N_{i,2} + 1 - N_{i,1} \\
a_2 & \quad\text{otherwise}
\end{cases}
\end{equation}

Note that, based on our voter model, the true winner is $a_1$. Thus, any voter $i$ is surprised if $Z_i = a_2$. Also, note the influence of the Beta prior parameters on a voter's eventual estimate. 
Plugging in different values of $\alpha$ and $\beta$, we can estimate the influence of traditional media on the surprise. We explain that in the following section.

\subsection{Traditional media model}
As detailed above, the impact of traditional media is represented through a Beta prior $B(\alpha, \beta) = \frac{\Gamma(\alpha + \beta) x^{\alpha - 1} (1-x)^{\beta - 1}}{\Gamma(\alpha) \Gamma(\beta)}$, with parameters $\alpha$ and $\beta$. For the purposes of our model, $\frac{\alpha - 1}{\alpha + \beta - 2}$ can be understood as the media's prediction of the probability of a voter preferring candidate $a_1$ over $a_2$.  Let $\alpha + \beta = t$, with $t$ representing the influence of traditional media on the individual voters. 
We model the parameter values as the true distribution, with a slight bias $\delta$, with $\alpha = t(\frac{1}{2} + (\epsilon - \delta))$ and $\beta = t(\frac{1}{2} - (\epsilon - \delta))$, with $\delta \in [-(\frac{1}{2} - \epsilon), \frac{1}{2} + \epsilon]$. Thus, $\delta$ can be thought of as the media bias towards the losing candidate $a_2$, with a large positive $\delta$ denoting a significant media bias towards the losing candidate, and a large negative $\delta$ denoting a significant bias towards the winning candidate. 

\section{Theoretical analysis}
We theoretically analyse the surprise of voters, in the case of an influential (but possibly biased) media source, followed by the case of a non-influential media source. These cases provide a basis for comparison between elections in which the voting population has a significant social media presence and cheap internet access and, consequently, a lower appetite for traditional media, and elections where traditional media is a major influence. 

\subsection{Influential media - Results}
In this theoretical analysis of our model, we analyse the surprise of the majority voters (voters of the winning candidate) and the minority voters (voters of the losing candidate) separately. In our model, we have used a boolean notion of surprise, where a voter is surprised if their predicted winner is not the winner of the election, and unsurprised otherwise. Note that, in this analysis, it does not matter if the voter's preferred candidate wins the election. If the voter had predicted that their preferred candidate's opponent would win, and their preferred candidate wins the election, they are still (pleasantly) surprised. However, the choice of the boolean notion of surprise, coupled with the social network structure of the voters, which is justifiably taken to be a filter bubble, implies that, in the absence of any traditional media, the majority voters would be unsurprised with a very high probability, while the minority voters would be surprised with a very high probability. Thus, we do separate evaluations for majority and minority voters with respect to the impact of traditional media on their surprise. Our modelling choice of a Beta prior for the impact of traditional media, with $\alpha, \beta$ representing the true distributions with some bias $\delta$ towards the losing candidate, as well as global weight $c$, also impacts the fraction of surprised people amongst the majority and minority voters differently. 

Our first theorem explains the impact of traditional media on the surprise of majority voters ($i \in H_1$). Since the majority voters are unsurprised with a high probability in the absence of traditional media, a well-established filter bubble ensures that a media source biased towards the losing candidate does not alter their prediction to the losing candidate. On the other hand, a weak filter bubble coupled with a media source biased towards the losing candidate can cause all of the majority voters to be surprised.

Our second theorem explains the impact of traditional media on the surprise of minority voters ($i \in H_2$). Since the minority voters are surprised with a high probability in the absence of traditional media, a well-established filter bubble only increases their surprise. Thus, the presence of a well-established filter bubble can cause the minority voters to be surprised with high probability. On the other hand, a weak filter bubble coupled with an influential media source can cause minority voters to be unsurprised with high probability.  
\begin{theorem}
	\label{Majority1}
If $\alpha = cn(\frac{1}{2} + (\epsilon - \delta))$ and $\beta = cn(\frac{1}{2} - (\epsilon - \delta))$ for some constant $c$, for sufficiently large $n$:   
\begin{itemize}
\item If $p > \frac{q(1-2\epsilon) - 4c(\epsilon - \delta)}{1 + 2\epsilon}$, for each voter $i \in H_1$, $P(Z_i = a_2) \leq e^{-2\sqrt{n-1}}$ i.e. the voter is surprised with a very low probability. Additionally, $P(\cap_{i \in H_1} Z_i = a_1) \geq 1 - (\frac{1}{2} + \epsilon)ne^{-2\sqrt{n-1}}$. 
\item If $p < \frac{q(1-2\epsilon) - 4c(\epsilon - \delta)}{1 + 2\epsilon}$, for each voter $i \in H_1$, $P(Z_i = a_2) \geq 1 - 2e^{-2\sqrt{n-1}}$ i.e. the voter is surprised with a very high probability. Additionally, $P(\cap_{i \in H_1} Z_i = a_2) \geq 1 - (1 + 2\epsilon)ne^{-2\sqrt{n-1}}$.
\end{itemize}  
\end{theorem}

\begin{proof}
By equation \ref{jointEq}, and with $i \in H_1$,

\begin{equation*}
(\alpha - \beta) \leq N_{i,2} - (N_{i,1} + 1) \implies Z_i =  a_2
\end{equation*}

Let $Y_i = N_{i,2} - N_{i,1}$. Note that, 
\begin{equation*}
Y_i = \sum_{j \in H_2} I((i,j) \in E) - \sum_{j \in H_1 - \{i\}}I((i,j) \in E)
\end{equation*}

Therefore, 
\begin{equation*}
E[Y_i] = (n(1/2 - \epsilon))q - (n(1/2 + \epsilon)-1)p
\end{equation*}

By Hoeffding's inequality,
\begin{equation}
\label{HGeq}
P(Y_i - E[Y_i] \geq t) \leq e^{-2t^2/(n-1)}
\end{equation}
\begin{equation}
\label{HLess}
P(|Y_i - E[Y_i]| < t) \geq 1 - 2e^{-2t^2/(n-1)}
\end{equation}

Let $X_i = Y_i - (2cn(\epsilon - \delta) + 1)$. 

Therefore, $P(X_i - E[X_i] \geq t) \leq e^{-2t^2/(n-1)}$ and $P(X_i - E[X_i] < t) \geq 1 - 2e^{-2t^2/(n-1)}$.

However, note that, if $X_i \geq 0$, $Z_i = a_2$ i.e. voter $i$ makes an incorrect prediction. 

\begin{equation*}
\begin{split}
E[X_i] & = (n(1/2 - \epsilon))q - (n(1/2 + \epsilon)-1)p - (2cn(\epsilon - \delta) + 1) \\
 & = n/2(q - (p + 2p\epsilon + 2q\epsilon + 4c(\epsilon - \delta))) - (1-p)
\end{split}
\end{equation*}

Let $t = (n-1)^{3/4}$. Note that, for sufficiently large $n$, $E[X_i] < 0$ if  $p > \frac{q(1-2\epsilon) - 4c(\epsilon - \delta)}{1 + 2\epsilon}$ and 

$E[X_i] > 0$ if $p <\frac{q(1-2\epsilon) - 4c(\epsilon - \delta)}{1 + 2\epsilon}$. \\

Note that, if $p > \frac{q(1-2\epsilon) - 4c(\epsilon - \delta)}{1 + 2\epsilon}$, $\exists n_1$ such that $E[X_i] + t < 0$ $\forall n \geq n_1$. 

Therefore, by equation \ref{HGeq}, 

$P(X_i \geq 0) \leq P(X_i > E[X_i] + t) = P(X_i > E[X_i] + (n-1)^{3/4}) \leq e^{-2\sqrt{n-1}}$ 

$\therefore P(Z_i = a_2) \leq e^{-2\sqrt{n-1}}$ 

By the union bound, 
\begin{equation*}
P(\cap_{i \in H_1} Z_i = a_1) = 1 - P(\cup_{i \in H_1} Z_i = a_2) \geq 1 - ((\frac{1}{2} + \epsilon)ne^{-2\sqrt{n-1}})
\end{equation*}

On the other hand, if $p < \frac{q(1-2\epsilon) - 4c(\epsilon - \delta)}{1 + 2\epsilon}$, $\exists n_2$ such that $E[X_i] - t > 0$ $\forall n \geq n_2$. 

Therefore, by equation \ref{HLess},

$P(X_i \geq 0) \geq P(|X_i - E[X_i]| < t) = P(|X_i - E[X_i]| < (n-1)^{3/4}) \geq 1 - 2e^{-2\sqrt{n-1}}$ 

$\therefore P(Z_i = a_2) \geq 1 - 2e^{-2\sqrt{n-1}}$

$\therefore P(Z_i = a_1) \leq 2e^{-2\sqrt{n-1}}$

By the union bound, 
\begin{equation*}
P(\cap_{i \in H_1} Z_i = a_2) = 1 - P(\cup_{i \in H_1} Z_i = a_1) \geq 1 - (1 + 2\epsilon)ne^{-2\sqrt{n-1}}
\end{equation*}

Hence, proved. \\
\end{proof}

\begin{theorem}
	\label{Minority1}
If $\alpha = cn(\frac{1}{2} + (\epsilon - \delta))$ and $\beta = cn(\frac{1}{2} - (\epsilon - \delta))$ for some constant $c$, for sufficiently large $n$: 
\begin{itemize}
\item If $p < \frac{q(1+2\epsilon) + 4c(\epsilon - \delta)}{1-2\epsilon}$, for each voter $i \in H_2$, $P(Z_i = a_2) \leq e^{-2\sqrt{n-1}}$ i.e. the voter is surprised with a very low probability. Additionally, $P(\cap_{i \in H_2} Z_i = a_1) \geq 1 - (\frac{1}{2} - \epsilon)ne^{-2\sqrt{n-1}}$.
\item If $p > \frac{q(1+2\epsilon) + 4c(\epsilon - \delta)}{1-2\epsilon}$, for each voter $i \in H_2$, $P(Z_i = a_2) \geq 1 - 2e^{-2\sqrt{n-1}}$ i.e. the voter is surprised with a very high probability. Additionally, $P(\cap_{i \in H_2} Z_i = a_2) \geq 1 - (1 + 2\epsilon)ne^{-2\sqrt{n-1}}$.
\end{itemize}
\end{theorem}
The proof idea is similar to that used in theorem \ref{Majority1}. We define a random variable $Y_i$ in terms of voter $i$'s neighbours in the social network graph. We use Hoeffding's inequality to bound the probability of $Y_i$ being outside a certain bound $t$ (equations \ref{HGeq}, \ref{HLess}). We then define $X_i = Y_i - ((2cn(\epsilon - \delta) - 1))$, and compute $E[X_i]$. For large $n$, $E[X_i] < 0$ if $p < \frac{q(1+2\epsilon) + 4c(\epsilon - \delta)}{1-2\epsilon}$ and  $E[X_i] > 0$ if $p > \frac{q(1+2\epsilon) + 4c(\epsilon - \delta)}{1-2\epsilon}$. If $X_i \geq 0$, $Z_i = a_2$. Setting $t = (n-1)^{3/4}$ in equations \ref{HGeq} and \ref{HLess}, we get the results mentioned above. Using the union bounds, we get the aggregate results for $H_2$. \\

\begin{corollary}
	\label{All1}
If $\alpha = cn(\frac{1}{2} + (\epsilon - \delta))$ and $\beta = cn(\frac{1}{2} - (\epsilon - \delta))$ for some constant $c$, if $p \in (\frac{q(1-2\epsilon) - 4c(\epsilon - \delta)}{1 + 2\epsilon}, \frac{(q(1+2\epsilon) + 4c(\epsilon - \delta))}{1-2\epsilon})$, the probability of \textbf{all} of the voters being unsurprised is greater than equal to $(1 - ne^{-2\sqrt{n-1}})$ for sufficiently large $n$.
\end{corollary}
The above corollary follows from the union bounds used in theorems \ref{Majority1} and \ref{Minority1}. 
The following are some interesting conclusions which can be drawn from these theorems:
\begin{enumerate}
\item For a certain range of $p$, (specified by corollary \ref{All1}) all of the voters will be unsurprised with a high probability. However, this range could be a null set, if $\delta > \epsilon$, i.e. the media bias towards the losing candidate is greater than the margin of victory. In such a case, if $p \in (\frac{(q(1+2\epsilon) + 4c(\epsilon - \delta))}{1-2\epsilon}, \frac{q(1-2\epsilon) - 4c(\epsilon - \delta)}{1 + 2\epsilon})$, all of the voters will be surprised with a high probability. (by theorems \ref{Majority1} and \ref{Minority1}).
\item Deeply entrenched filter bubbles ($p$ much larger than $q$) reduce the surprise of the  winning candidate voters, while increasing the surprise of the losing candidate voters. 
\item An influential media source which is neutral or biased towards the winning candidate ($\delta \leq 0$) reduces the surprise of all of the voters.
\item An influential media source (whose bias towards the losing candidate is less than the margin of victory ($0 < \delta < \epsilon$)) has a positive influence on reducing the surprise of all of the voters.  
\item If the margin of victory in the election is huge, all of the voters are less likely to be surprised. (From corollary \ref{All1}, it is easy to observe that $p$'s range is wider in case of an election with a large margin of victory.) 
\end{enumerate}

\subsection{Uninfluential media - Results}
We model an uninfluential media source by setting $t = a n^{\gamma}$, where $\gamma, a$ are some constants with $\gamma \in (0, 1)$. We first look at the impact of uninfluential media on the surprise of the majority voters followed by the effect on the surprise of the minority voters for sufficiently large $n$. However, it is important to note that all of these results hold only for sufficiently large $n$, and are not applicable in the cases of smaller elections. We augment this deficiency by performing an experimental analysis on a small sample of voters in the Brexit dataset, to understand the impact of traditional media in the case of smaller elections.   
\begin{theorem}
	\label{Majority2}
	If $\alpha = an^{\gamma}(\frac{1}{2} + (\epsilon - \delta))$ and $\beta = an^{\gamma}(\frac{1}{2} - (\epsilon - \delta))$ for some constants $a, \gamma$ with $\gamma \in (0,1)$, for sufficiently large $n$, for each voter $i \in H_1$, $P(Z_i = a_2) \leq e^{-2\sqrt{n-1}}$. Additionally, $P(\cap_{i \in H_1} Z_i = a_1) \geq 1 - (\frac{1}{2} + \epsilon)ne^{-2\sqrt{n-1}}$.
\end{theorem} 
\begin{proof}
	By equation \ref{jointEq}, and with $i \in H_1$,
	\begin{equation*}
	(\alpha - \beta) \leq N_{i,2} - (N_{i,1} + 1) \implies Z_i =  a_2
	\end{equation*}
	
	Let $Y_i = N_{i,2} - N_{i,1}$. Note that, 
	\begin{equation*}
	Y_i = \sum_{j \in H_2} I((i,j) \in E) - \sum_{j \in H_1 - \{i\}}I((i,j) \in E)
	\end{equation*}
	
	Therefore, 
	\begin{equation*}
	E[Y_i] = (n(1/2 - \epsilon))q - (n(1/2 + \epsilon)-1)p
	\end{equation*}
	
	By equation \ref{HGeq},
	
	$P(Y_i - E[Y_i] \geq t) \leq e^{-2t^2/(n-1)}$ 
	
	Let $X_i = Y_i - (2an^{\gamma}(\epsilon - \delta) + 1)$. 
	
	$\therefore P(X_i - E[X_i] \geq t) \leq e^{-2t^2/(n-1)}$. However, note that, if $X_i \geq 0$, $Z_i = a_2$ i.e. voter $i$ makes 
	
	the incorrect prediction. 
	\begin{equation*}
	\begin{split}
	E[X_i] & = (n(1/2 - \epsilon))q - (n(1/2 + \epsilon)-1)p - (2an^{\gamma}(\epsilon - \delta) + 1) \\ 
	& = n/2(q - (p + 2p\epsilon + 2q\epsilon + 4an^{\gamma - 1}(\epsilon - \delta))) - (1-p)
	\end{split}
	\end{equation*}

	Note that, for sufficiently large $n$, $E[X_i] < 0$ if $q < \frac{p(1 + 2\epsilon)}{1 - 2\epsilon}$ or equivalently, $p > \frac{q(1-2\epsilon)}{1 + 2\epsilon}$.
	
	($\because$ $n^{\gamma - 1} \rightarrow 0$ as $n \rightarrow \infty$). 
	
	Let $t = (n-1)^{3/4}$. Note that,  $\exists n_1$ such that $E[X_i] + (n-1)^{3/4} < 0$ $\forall n \geq n_1$.
	
	However, $p > q \implies p > \frac{q(1-2\epsilon)}{1 + 2\epsilon}$.

	$\therefore P(X_i \geq 0) \leq P(X_i > E[X_i] + (n-1)^{3/4}) \leq e^{-2\sqrt{n-1}}$. 
	
	$\therefore P(Z_i = a_2) \leq e^{-2\sqrt{n-1}}$.
	
	Using union bound, 
	\begin{equation*}
P(\cap_{i \in H_1} Z_i = a_1) = 1 - P(\cup_{i \in H_1} Z_i = a_2) \geq 1 - (\frac{1}{2} + \epsilon)ne^{-2\sqrt{n-1}}
	\end{equation*}
	
	Hence, proved. \\
\end{proof}

The above theorem states that any voter of the winning candidate is going to be unsurprised with high probability. Additionally, it states the stronger result that all of the voters of the winning candidate will be unsurprised with high probability.

\begin{theorem}
	\label{Minority2}
If $\alpha = an^{\gamma}(\frac{1}{2} + (\epsilon - \delta))$ and $\beta = an^{\gamma}(\frac{1}{2} - (\epsilon - \delta))$ for some constants $a, \gamma$ with $\gamma \in (0,1)$, for sufficiently large $n$, for each voter $i \in H_2$, $P(Z_i = a_2) \geq 1 - 2e^{-2\sqrt{n-1}}$. Additionally, $P(\cap_{i \in H_2} Z_i = a_2) \geq 1 - (1 + 2\epsilon)ne^{-2\sqrt{n-1}}$.
\end{theorem}
The above theorem states that any voter of the losing candidate is going to be surprised with high probability. It also states the stronger result that all of the voters of the losing candidate will be surprised with high probability. 

The proof idea is fairly similar to that of theorems \ref{Majority1} and \ref{Minority1}. The only change is that the influence of traditional media is reduced, so, there is an $n^{\gamma}$ factor representing traditional media instead of $n$, with $\gamma \in (0, 1)$. As shown in the proof of theorem \ref{Majority2}, the $n^{\gamma}$ term can be ignored for significantly large $n$. By using the Hoeffding bound as in theorem \ref{Majority1}, it can be shown that any minority voter will be surprised with a high probability.  Using the union bound on the first result, we get the stronger result of all of the voters of the losing candidate being surprised with high probability. \\ \\
Thus, our uninfluential media model reduces to the case of a complete absence of traditional media for both majority and minority voters. Thus, regardless of the values of $p$ and $q$, as long as $p > q$,  we have the majority voters being unsurprised with a very high probability, and the minority voters being surprised with a very high probability. However, these results are for elections with a huge number of voters, and do not necessarily follow in the case of smaller elections. For an analysis of smaller elections, we conduct experiments on small sub-samples of the Brexit dataset. 

\section{Experimental analysis}
Our results so far have been theoretical. In this section, we validate our theoretical results by performing experiments on the Brexit dataset\footnote{\url{https://goo.gl/xBWA5q}}. We modify our theoretical model very slightly to account for the geographical aspects of the data. Most of our experimental analysis choices mirror those described in \citet{dey2018surprise}.
\subsection{Experimental model}

\subsubsection{Voter model}
The voter model remains the same, with one key difference. Since the Brexit dataset is fairly large, we randomly sample a set of 10000 voters from the total list of voters. We treat this set of voters as the voters for this ``sub-election'', with all of the results being for this sub-election.
\subsubsection{Voter connection model} 
The Brexit dataset consists of 33 million valid votes, across 382 regions in the UK, who voted either remain (R) or leave (L). We use this geographical information, by finding out the latitude and longitude of these regions from another dataset\footnote{\url{https://goo.gl/BBPX4j}}. We then treat each of the voters from a region as residing at that particular latitude and longitude. We need to decide the connection probabilities between voters to complete this model. The following are reasonable requirements from this connection probability model:
\begin{itemize}
\item Voters who are nearby should have a larger connection probability than those farther away. Specifically, for three voters $x$, $y$ and $z$, if $y$ and $z$ belong to the same class, and the distance between $x$ and $y$ is less than the distance between $x$ and $z$, the connection probability of $x$ and $y$ should be higher than that of $x$ and $z$. 
\item Voters who are from the same class should have a higher connection probability than those from opposite classes. Specifically, for three voters $x$, $y$ and $z$, if $y$ and $z$ are at the same location, but belong to different classes, and, without loss of generality, $x$ and $y$ belong to the same class, the connection probability of $x$ and $y$ should be higher than that of $x$ and $z$.  
\end{itemize}  Using the position (latitude and longitude) of a voter, we weight the connection probability of two voters. Thus, the connection probability of voters $i$ and $j$, 
\begin{equation}
r_{i,j} = r_{j,i} = p_{class} \times p_{dist}
\end{equation}  
where 
\begin{equation}
p_{class} = \begin{cases}
p &\quad\text{if } i,j \in H_k,  k \in \{1,2\} \\
q & \quad\text{otherwise}
\end{cases} 
\end{equation}
and 
\begin{equation}
p_{dist} = \exp^{-\text{weight} \times \text{distance}}
\end{equation}
with $\text{weight} = 0.1$, $r_{i,j}$ denoting the probability of $i$ being able to observe $j$'s preferred candidate, and $r_{j,i}$ denoting the probability of $j$ being able to observe $i$'s preferred candidate. 
Each of the voters in the sampled set tries to connect to each of the other voters, with their connection probabilities as described above.
\subsubsection{Voter's winner perception model} 
Each voter's winner perception model remains the same as in the theoretical case, with traditional media sources being modelled as the Beta prior, and their social connections being the data or evidence for a binomial distribution. Each voter $i$ estimates the unknown parameter $q_1$ by using the MAP estimate and coming up with their own (inaccurate) estimate $q_{1,i}$. 

\subsubsection{Traditional media model}
The impact of traditional media is represented through a Beta prior, $B(\alpha, \beta)$, with $\alpha + \beta = cn$ in the case of influential media, where $n$ is the number of voters and the constant $c$ denotes the \textit{global weight}. On the other hand, in the case of uninfluential media, $\alpha + \beta = an^{\gamma}$ where $\gamma$ is the measure of influence and $a$ is the \textit{global weight}. Just like our theoretical results, we assume that these parameter values are in the ratio of the true distribution, with a bias $\delta$ towards the losing candidate with $\delta \in [-(\frac{1}{2} - \epsilon), \frac{1}{2} + \epsilon]$.  Thus, the parameter values are $\alpha = t(\frac{1}{2} + (\epsilon - \delta))$ and $\beta = t(\frac{1}{2} - (\epsilon - \delta))$

\subsection{Results - Influential media}
\begin{figure}[H]
	\centering
	\begin{subfigure}[b]{0.47\textwidth}
		\includegraphics[width=\textwidth, height= 0.21 \textheight]{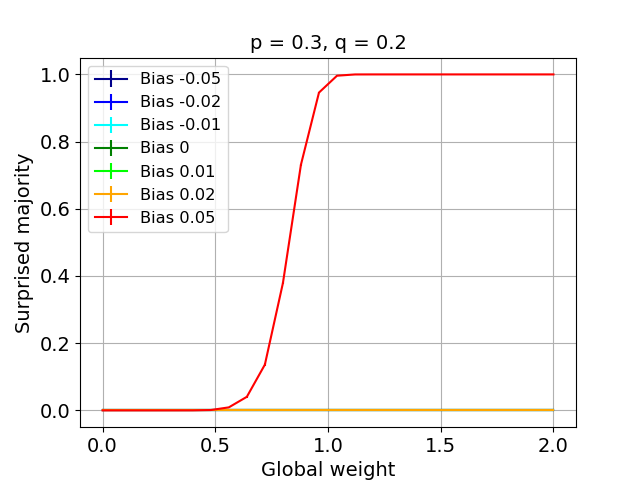}
	\end{subfigure}
	\begin{subfigure}[b]{0.47\textwidth}
		\includegraphics[width=\textwidth,  height= 0.21 \textheight]{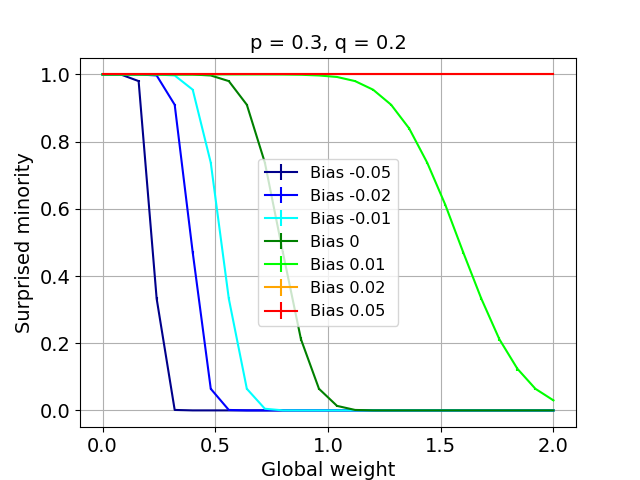}
	\end{subfigure}
	\begin{subfigure}[b]{0.47\textwidth}
		\includegraphics[width=\textwidth,  height= 0.21 \textheight]{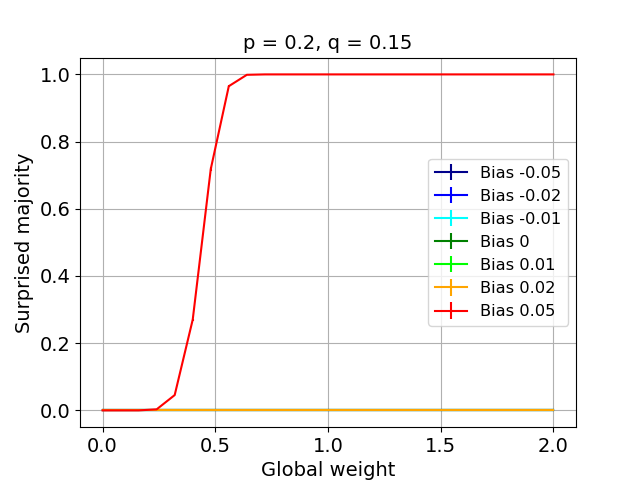}
	\end{subfigure}
	\begin{subfigure}[b]{0.47\textwidth}
		\includegraphics[width=\textwidth,  height= 0.21 \textheight]{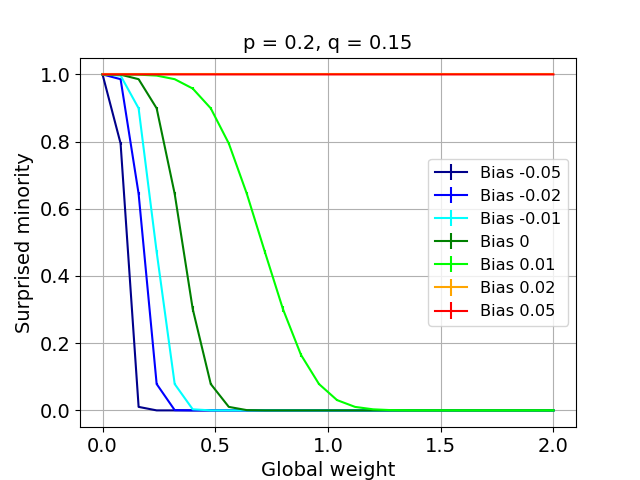}
	\end{subfigure}

	\begin{subfigure}[b]{0.48\textwidth}
		\includegraphics[width=\textwidth,  height= 0.21 \textheight]{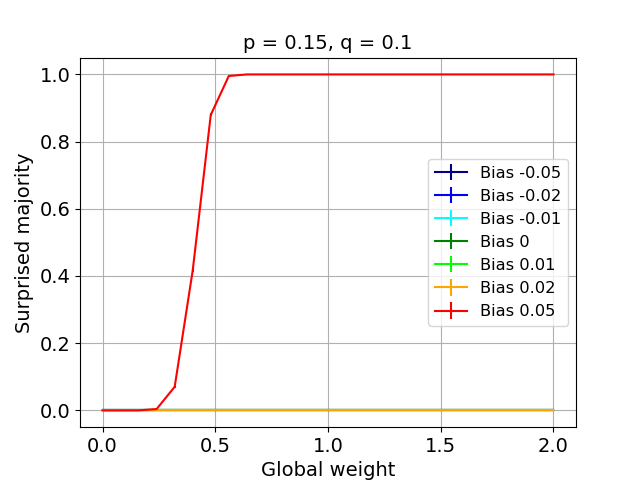}
	\end{subfigure}
	\begin{subfigure}[b]{0.48\textwidth}
		\includegraphics[width=\textwidth,  height= 0.21 \textheight]{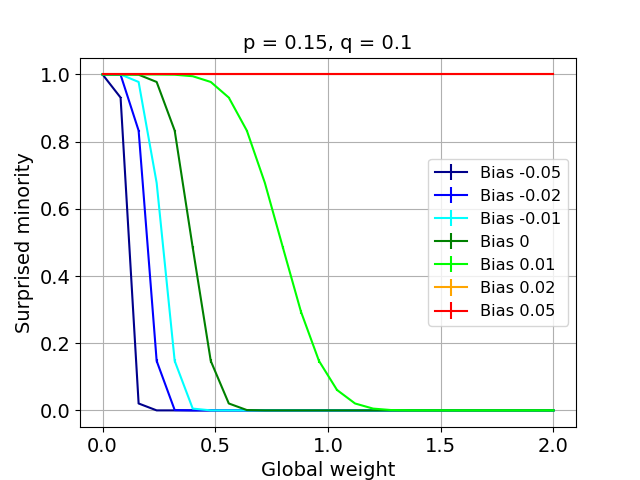}
	\end{subfigure}
	\begin{subfigure}[b]{0.48\textwidth}
		\includegraphics[width=\textwidth,  height= 0.21 \textheight]{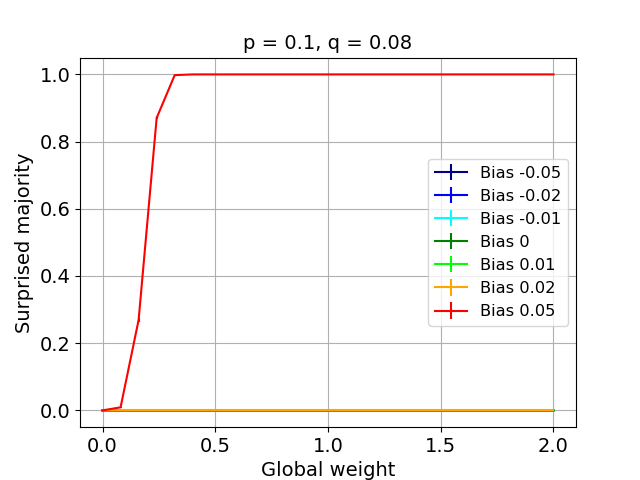}
	\end{subfigure}
	\begin{subfigure}[b]{0.48\textwidth}
		\includegraphics[width=\textwidth,  height= 0.21 \textheight]{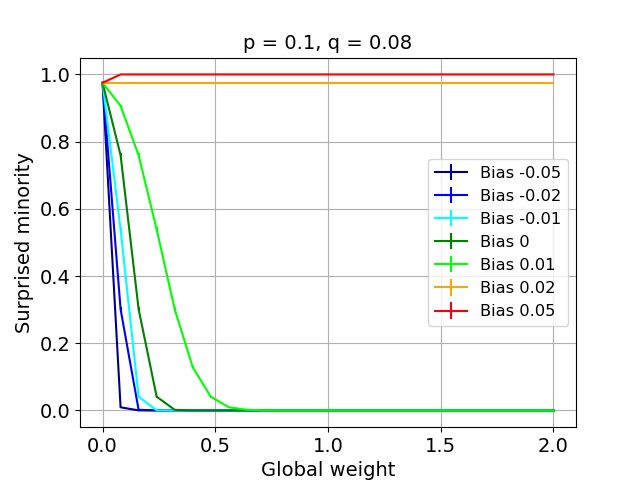}
	\end{subfigure}
	\captionsetup{labelformat=empty}
	\caption{Figure 1. Fraction of surprised majority/minority voters for different values of connection probabilities $p$, $q$ and bias $\delta$, with global weight $c = \frac{\alpha + \beta}{n}$. The voter sample has 4800 remain voters, and 5200 leave voters. }\label{fig:PM}
\end{figure}

\subsection{Results - Uninfluential media}
\begin{figure}[H]
	\centering
	\textbf{$\gamma = 0.5$}\par\medskip
	\begin{subfigure}[b]{0.47\textwidth}
		\includegraphics[width=\textwidth, height= 0.19 \textheight]{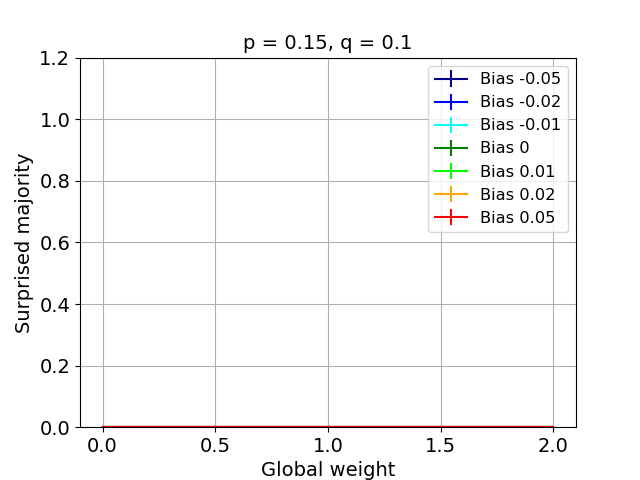}
	\end{subfigure}
	\begin{subfigure}[b]{0.47\textwidth}
		\includegraphics[width=\textwidth,  height= 0.19 \textheight]{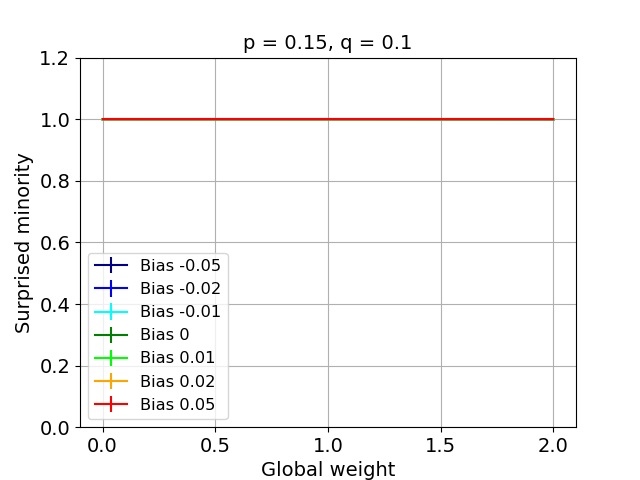}
	\end{subfigure}
	\begin{subfigure}[b]{0.47\textwidth}
		\includegraphics[width=\textwidth,  height= 0.19 \textheight]{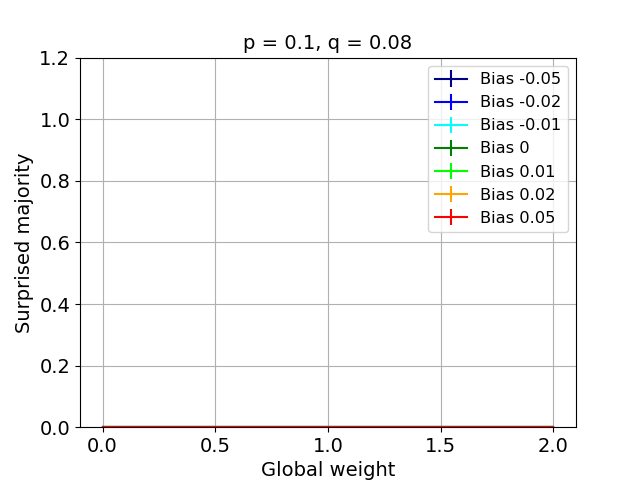}
	\end{subfigure}
	\begin{subfigure}[b]{0.47\textwidth}
		\includegraphics[width=\textwidth,  height= 0.19 \textheight]{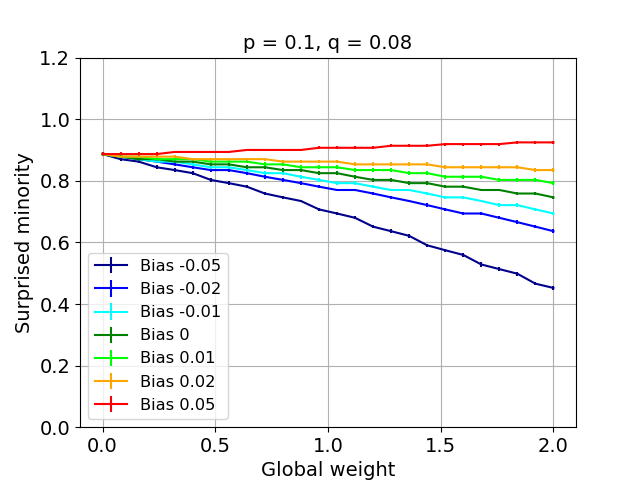}
	\end{subfigure}
\end{figure}
\begin{figure}[H]
	\centering
	\textbf{$\gamma = 0.8$} \par\medskip
	\begin{subfigure}[b]{0.48\textwidth}
		\includegraphics[width=\textwidth,  height= 0.19 \textheight]{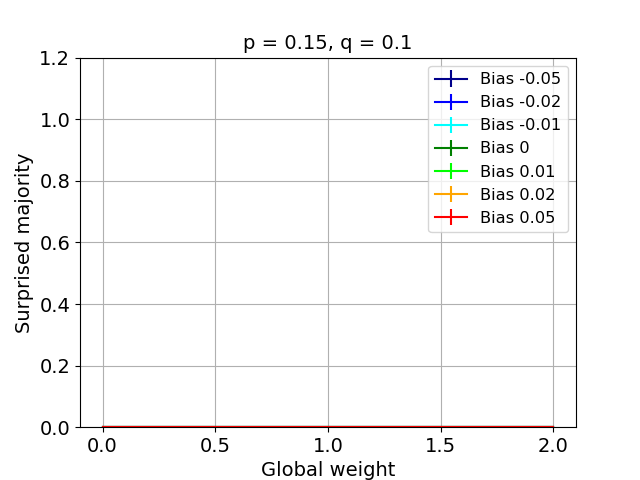}
	\end{subfigure}
	\begin{subfigure}[b]{0.48\textwidth}
		\includegraphics[width=\textwidth,  height= 0.19 \textheight]{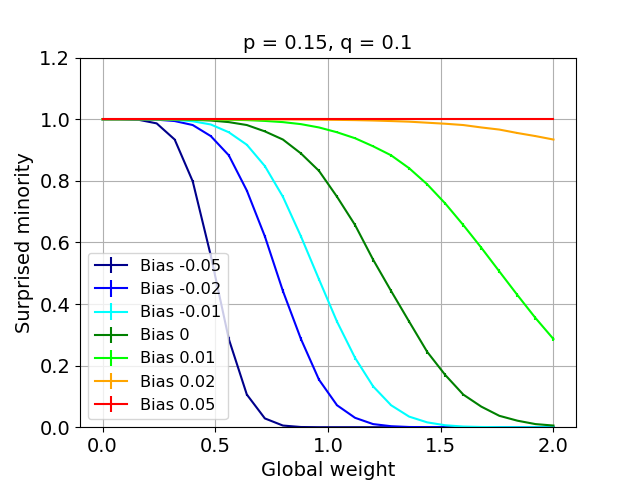}
	\end{subfigure}
	\begin{subfigure}[b]{0.48\textwidth}
		\includegraphics[width=\textwidth,  height= 0.19 \textheight]{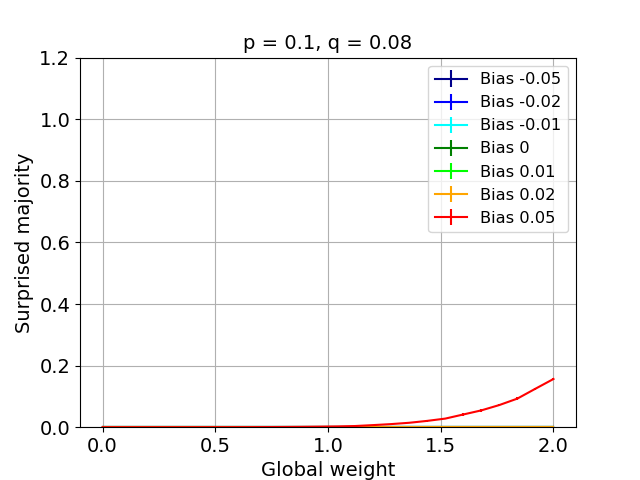}
	\end{subfigure}
	\begin{subfigure}[b]{0.48\textwidth}
		\includegraphics[width=\textwidth,  height= 0.19 \textheight]{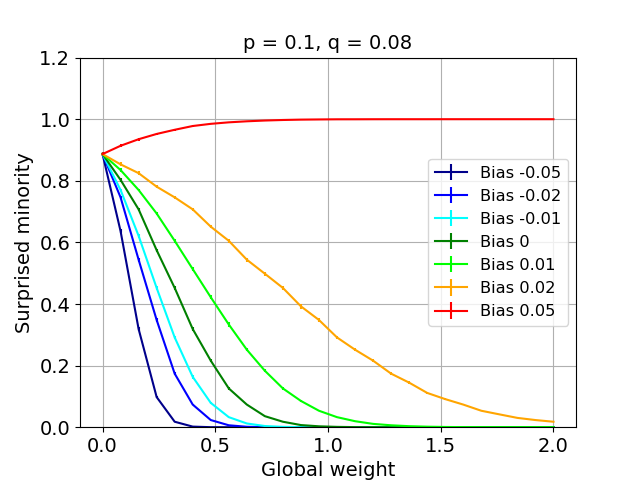}
	\end{subfigure}
	\captionsetup{labelformat=empty}
	\caption{Figure 2. Fraction of surprised majority/minority voters for different values of connection probabilities $p$, $q$, bias $\delta$, and measure of influence $\gamma$ for the case of uninfluential media with global weight $a = \frac{\alpha + \beta}{n^{\gamma}}$. The voter sample has 4666 remain voters, and 5334 leave voters. }
\end{figure}

\subsection{Analysing the results}
The majority voters denote the voters of the winning candidate, (in this case, the leave resolution) while the minority voters denote the voters of the losing candidate (in this case, the remain resolution). Additionally, as mentioned above, the bias $\delta$ denotes the bias of the media towards the losing candidate while $p$ and $q$ are the intra-class and inter-class connection probabilities of the voters. The y-axis denotes the fraction of the majority/minority voters who are surprised. 
\subsubsection{Influential media}
The results seem to follow the theoretical predictions fairly accurately. In the case of influential media, $\alpha + \beta = cn$ where $n$ is the total number of voters in the sub-election, while $c$ is the global weight (represented on the x-axis) denoting the impact of traditional media on the election. Interesting points to be noted include:
\begin{itemize}
	\item For the majority voters, a large difference between $p$ and $q$ (implying a deeply ingrained filter bubble) would result in less surprise in the case of a media heavily biased towards the losing candidate. This can be clearly seen from the majority voter graphs with $p = 0.3, q = 0.2$ and with $p = 0.1, q = 0.08$ respectively. 
	\item For the minority voters, a deeply ingrained filter bubble would result in more surprise. This can be clearly seen from the minority voter graphs with $p = 0.3, q = 0.2$ and with $p = 0.1, q = 0.08$ respectively.
	\item A media heavily biased towards the winning candidate is going to do better at reducing surprise than a more neutral media source.  
\end{itemize}

\subsubsection{Uninfluential media}
In the case of uninfluential media, $\alpha + \beta = an^{\gamma}$ where $n$ is the total number of voters in the sub-election, $a$ is the global weight (represented on the x-axis) denoting the impact of traditional media on the election, while $\gamma \in (0, 1)$ is the influence measure. The results are slightly different from the theoretically predicted results. However, note that, in our theoretical analysis, we had discarded the $n^{\gamma}$ term for significantly large $n$. In our experimental analysis, the term becomes significant. The following are some of the interesting insights obtained from the experiments:
\begin{itemize}
\item In the case of a heavily entrenched filter bubble ($p$ much larger than $q$), and a small influence parameter $\gamma$, the majority voters will be unsurprised with high probability and the minority voters will be surprised with high probability. This can be noted from the figures for surprised majority/minority for $\gamma = 0.5$, with $p = 0.15, q = 0.1$.
\item If the influence parameter $\gamma$ is large, the minority voters will be less surprised in the case of a traditional media which is neutral or biased towards the winner. This effect will be more prominent in the case of a weak filter bubble. 
\end{itemize} 

\section{Conclusion}
In this work, we have introduced an aggregate notion of surprise (based on the boolean model of surprise introduced in \citet{dey2018surprise}) which deals with the fraction of surprised people based on how close the election is, the presence/absence of an influential media source and the influence of filter bubbles on their social connections. We have shown that, in the presence of an influential media source which is neutral or biased towards the winning candidate, the effects of filter bubbles can be curbed to a certain extent. Additionally, influential media sources biased towards the losing candidate can be a major cause for surprising elections.

A good area for further research is to look at a non-boolean model of surprise using some margin of victory measure to figure out the surprise associated with an election. A straightforward way to do this would be to use a distance measure between a voter's predicted margin of victory and the actual margin of victory, and using an aggregation measure (such as the mean of the absolute value of the difference between predicted margin of victory and actual margin of victory). We could then try to find out which media distribution would lead to the least surprised voters.

Another direction requiring study is the extension to elections with more than two candidates. This extension is difficult because, in elections with more than two candidates, there is a dependency on the particular anonymised voting rule being used. However, such a result, even for limited voting rules, would improve our understanding of surprise in elections significantly. 

	\bibliographystyle{plainnat}
	\bibliography{mybibliography}
	%
	
	%
	%
	%
	%
\end{document}